\documentclass[a4paper,USenglish,cleveref,thm-restate]{lipics-v2021}

\pdfoutput=1
\usepackage{style}
\usepackage{shortcuts}

\title{On Extremal Properties of \texorpdfstring{\boldmath$k$}{k}-CNF:\texorpdfstring{\\}{ }Capturing Threshold Functions}
\titlerunning{On Extremal Properties of \boldmath$k$-CNF: Capturing Threshold Functions}
\author{Mohit Gurumukhani}{Cornell University, Ithaca, USA \and \url{https://www.mohitgurumukhani.com/}}{mgurumuk@cs.cornell.edu}{}{}
\author{Marvin K\"{u}nnemann}{Karlsruhe Institute of Technology, Germany \and \url{https://act.iti.kit.edu/people/marvinkuennemann}}{marvin.kuennemann@kit.edu}{}{}
\author{Ramamohan Paturi}{University of California San Diego, San Diego, USA \and \url{https://cseweb.ucsd.edu/~paturi/}}{rpaturi@ucsd.edu}{}{}
\authorrunning{M. Gurumukhani, M. K\"{u}nnemann and R. Paturi}

\Copyright{Mohit Gurumukhani, Marvin K\"{u}nnemann and Ramamohan Paturi}

\acknowledgements{We want to thank Nick Fischer for helpful discussions.}

\nolinenumbers
\hideLIPIcs

\begin{document}

\maketitle

\begin{abstract}
We consider a basic question on the expressiveness of $k$-CNF formulas: How well can $k$-CNF formulas capture threshold functions? Specifically, what is the largest number of assignments (of Hamming weight $t$) accepted by a $k$-CNF formula that only accepts assignments of weight at least $t$?

Among others, we provide the following results:
\begin{itemize}
\item While an optimal solution is known for $t \leq n/k$, 
the problem remains open for $t > n/k$. We formulate a (monotone) version of the problem 
as an extremal hypergraph problem and show that for $t = n-k$, the problem is exactly the Tur\'{a}n problem. 
\item For $t = \alpha n$ with constant $\alpha$, we provide a construction and show its optimality for $2$-CNF. Optimality of the construction for $k>2$ would give improved lower bounds for depth-$3$ circuits. 
\end{itemize}
\end{abstract}

\section{Introduction} \label{sec:introduction}

\newcommand{\THR}{\mathrm{THR}}
\newcommand{\poly}{\mathrm{poly}}

For the currently best known worst-case upper bounds for $k$-SAT,
the savings over exhaustive search decreases rapidly as $k$, the clause width, increases.
In fact, the best known algorithms run in time $2^{(1-\Theta(1/k))n}$ \cite{PaturiPudlakSaksZane_2005_jacm,
 HansenKaplanZamirZwick_2019_stoc, SchederTalebanfard_2020_ccc}.
The possibility that the savings of $\Theta(1/k)$ is the best possible is known as the Super Strong Exponential Time Hypothesis (see~\cite{VyasWilliams_2019_sat}). 

The situation is `dual' with respect to exponential size lower bounds on 
depth-3 circuits, specifically, $\Sigma\Pi\Sigma_k$ circuits (disjunctions of $k$-CNFs). The best known lower bounds for $\Sigma\Pi\Sigma_k$ circuits for any explicit function
are of the form $2^{\Omega(\frac{n}{k})}$ \cite{PaturiSaksZane_2000_cc,PaturiPudlakZane_1999_cjtcs}.
Establishing circuit lower bounds for depth-3 circuits has a rich history, particularly due to a classical result of Valiant \cite{Valiant_1977_mfcs}, that shows that establishing lower bounds for depth-3 circuits with bottom fan-in $O(n^{\epsilon})$ implies nonlinear size bounds on  logarithmic-depth circuits. Valiant also shows that a lower bound on $\Sigma\Pi\Sigma_k$ circuits of $2^{cn}$ for $c>0$ independent of $k$ implies a nonlinear size lower bound on series-parallel circuits \cite{Valiant_1977_mfcs,Calabro_2008_eccc}. 
For a recent overview and results in this direction, see \cite{GolovnevKulikovWilliams_2021_itcs} and references therein. 

Combinatorial properties of the solution 
spaces of $k$-CNFs played a critical role in the development of the best known upper bounds for  $k$-SAT and lower bounds for depth-3 circuits.
We believe that a deeper understanding of the  combinatorial properties of solution sets of $k$-CNF is a productive approach to make progress on $k$-SAT 
as well as lower bounds on depth-3 circuits.
\cite{PaturiSaksZane_2000_cc,PaturiPudlakZane_1999_cjtcs,GolovnevKulikovWilliams_2021_itcs}.
In this paper, we study a basic extremal property of the solution sets of $k$-CNF formulas, motivated by the question: How well can $k$-CNFs capture threshold functions? A \emph{threshold function} $\THR_t : \{0,1\}^n \to \{0,1\}$ is defined by $\THR_t(x_1, \dots, x_n) = 1$ if and only if $\sum_{i=1}^n x_i \ge t$. It is easy to see that not every $\THR_t$ can be written as a single $k$-CNF. We are interested in the following questions:
\begin{enumerate}
\item  What are the largest subsets of solutions of $\THR_t$  that can be captured by $k$-CNF (and monotone $k$-CNF)?
\item How many $k$-CNFs are required to express $\THR_t$ as a disjunction of $k$-CNFs? 
\end{enumerate}

Computation of threshold functions in various circuits models has been studied by several researchers and there are a number of results on the 
bounds for computing threshold functions.
In applications involving satisfiability solvers, researchers have studied how to encode threshold functions -- or, equivalently, \emph{cardinality constraints} -- as CNFs. In this context, however, there is usually no restriction on clause width, the objective is typically to minimize the formula size (essentially, the number of clauses) and the use of auxiliary variables is usually allowed (see, e.g., \cite{Sinz_2005_cp,Marques-SilvaLynce_2007_cp}). Note that such cardinality constraints occur frequently, e.g., when solving MaxSAT instances by multiple calls to a SAT solver. 

Several
researchers~\cite{Boppana_1984_stoc,KlawePaulPippengerYannakakis_1984_stoc,Radhakrishnan_1994_comb,HastadJuknaPudlak_1995_cc,Wolfovitz_2006_ipl}
studied the total size of optimal $\THR_t$ representations as a
disjunction of CNFs (i.e., $\Sigma \Pi \Sigma$ circuits). These results
show (non-matching) lower and upper bounds for the second question
(e.g., that $\THR_{n/2}$ requires size between $2^{\Omega(\sqrt{n})}$
and $2^{O(\sqrt{n \log n})}$), yet these results translate to only weak
lower bounds for the first question (see below for details).

There are also several results regarding the representation of threshold functions as (probabilistic) polynomials, which have led to significant progress for SAT algorithms (see particularly~\cite{AlmanChanWilliams_2016_focs}).
However, representation of threshold functions as (probabilistic) polynomials is very different from representing them as disjunctions of $k$-CNFs.

To formalize our questions, for $t\geq 0$ we define $S(n,t,k)$ as the largest number of assignments of Hamming weight $t$ accepted by an $n$-variable $k$-CNF formula which does not accept an assignment of Hamming weight less than $t$.\footnote{One could also define a variant of this quantity where we seek to determine the largest number of assignments of Hamming weight \emph{at least $t$}, instead of \emph{exactly} $t$. In a sense that will be made formal in \cref{obs:kCNFcovering}, accepting assignments of weight exactly $t$ presents the central difficulty, so we focus on this version.} We shall see below that this quantity tightly controls how many $k$-CNF formulas are needed to cover a threshold formula. It might appear surprising that this questions is not yet resolved. For $t\le n/k$ there is an immediate construction that is proven optimal by Wolfovitz~\cite{Wolfovitz_2006_ipl}: divide the variables into $t$ blocks $B_1, \dots, B_t$ of size $k$, and a block $B_0$ of the remaining $n-tk$ variables. For each block $B_i$, introduce the clause $(x_1\vee \cdots \vee x_k)$ where $x_1, ..., x_k$ are the variables in $B_i$, and for each variable $x$ in $B_0$ introduce the unit clause $\bar{x}$. It is obvious that this formula accepts only assignments of weight at least $t$, and accepts precisely $k^t$ assignments of weight $t$, which is the best possible~\cite[Lemma 1]{Wolfovitz_2006_ipl}.  However, for larger $t$, Wolfovitz only gives a $k^t$ upper bound (which becomes trivial e.g. when $t\ge n/\log k$), without matching constructions.\footnote{\label{footnote:Wolfovitz} Note that Wolfovitz implicitly claims this construction to be optimal for all values of $t$. However, once $t\ge n/k$, we can no longer obtain $t$ \emph{disjoint} blocks of $k$ variables, and indeed, it turns out that this setting is combinatorially much more challenging.}

\subsection{Our Results}

Our first result demonstrates the difficulty of determining $S(n,t,k)$ by establishing a  connection to a classic open problem in combinatorics, specifically, the Tur\'an problem. 
The \emph{Tur\'an number} $T(n,q,k)$ denotes the size of the smallest system $\mathcal{S}$ of size $k$ sets ($k$-sets) over $[n]$ such that every $q$-set over $[n]$ contains some $S\in \mathcal{S}$. 
It turns out that for large thresholds $t=n-k$, $S(n,t,k)$ is determined by a corresponding Tur\'an number:  

\begin{theorem}\label{thm:turan1}
	For any $n\ge k \ge 2$, we have that $S(n,n-k,k)= \binom{n}{k} - T(n,k+1,k)$.
\end{theorem}

We define the asymptotic Tur\'an numbers as $T(q,k) = \limsup T(n,q,k)/\binom{n}{k}$.
Unfortunately, the exact value is not known even for $T(4,3)$
(in fact, Erd\H{o}s posed a reward of \$500 for this question). 
It is known that $0.4383334 \leq T(4,3)\leq 4/9$ \cite{Keevash_2011}. However, 
closing the gap remains to be an open problem.
Thus determining $S(n,t,3)$ for all $t$ requires further progress on Tur\'an numbers.

Note that it is known that Tur\'an numbers have a close connection 
to CNF formula size 
(as opposed to the number of solutions) 
\cite{Sidorenko_1995_gc}.
We show that for the problem of determining $S(n,t,k)$  for $t=n-k$, 
minimizing the formula size is equivalent to maximizing the solution set (note that this is not necessarily true in general). 
\cref{thm:turan1}   demonstrates the underlying difficulty of determining $S(n,t,k)$ for all values of $t$ even for constant $k$. Similarly, for the regime of large clause widths (where $n-k$ is constant), we obtain a exact correspondence to Steiner systems, whose existence is extensively studied (\cref{sec:Steiner}).

Let us now consider the `middle' region, the case of  linear thresholds, i.e., $t = \alpha n$ for some $\alpha\in [0,1]$. An
obvious generalization of Wolfovitz's construction is to partition the variables into blocks of size at most $k$, and ensure that for each block, all satisfying assignment have at least $\alpha k$ ones. 
We show that this obvious construction is not optimal. We obtain a better construction by choosing block sizes adaptively based on the threshold. 

\noindent
\textbf{Adaptive block construction:}
For a clearer presentation, we assume that $b:=(k-1)/(1-\alpha)$ is an integer 
and it divides $n$. We divide the variables into blocks $B_1, \dots, B_{n/b}$ of size $b$. For each block $B_i$, we introduce monotone clauses $(x_1\vee x_2 \vee \cdots \vee x_k)$ where $x_1, \dots, x_k$ are distinct variables in $B_i$. It is straightforward to see (\cref{sec:bounds}) that in any satisfying assignment, each block $B_i$ must have at least $b-(k-1)$ variables set to $1$. Since by definition of $b$, we have $b-(k-1) = \alpha b$, each block indeed enforces an $\alpha$-fraction of its variables to be set to $1$. This yields a lower bound of
\[ S(n, \alpha n, k) \ge \binom{b}{\alpha b}^{n/b} \text{ whenever $b=(k-1)/(1-\alpha)$ divides $n$}. \]
Assume that $\alpha k$ is an integer and $k,b$ both divide $n$, then it follows from $b> k$ that this bound improves over $\binom{k}{\alpha k}^{n/k}$, the number of solutions of the obvious generalization of Wolfovitz's construction. We formulate the hypothesis that this construction is optimal:

\begin{conjecture}\label{conjecture}
	The adaptive construction is optimal: Let $k\ge 2, \alpha \in [0,1]$  such that $b =(k-1)/(1-\alpha)$ is an integer. Whenever $b$ divides $n$, we have
	\[  S(n, \alpha n, k) = \binom{b}{\alpha b}^{n/b}. \]
\end{conjecture}

\paragraph*{The $k=2$ case.}

While we could not establish \cref{conjecture} for $k\geq 3$, 
we provide evidence for \cref{conjecture} by showing that the 
adaptive block construction is indeed optimal for $k=2$. 

\begin{theorem}[Optimality for $k=2$, Informal Version]
	For every $\alpha \in [0,1]$ and $n$ such that $b:= 1/(1-\alpha)$ is an integer dividing $n$, we have 
	\[ S(n,\alpha n, 2) = \binom{b}{\alpha b}^{n/b} = \left(\frac{1}{1-\alpha}\right)^{n(1-\alpha)}. \]
\end{theorem}

In fact, we can establish the exact values of $S(n,t,2)$ for all $t$ -- note that for $k=2$, the corresponding Tur\'an number $T(n,3,2)$ is known. For larger clause width $k\ge 3$, it is likely that we need to exploit properties unique to the linear threshold regime to resolve our main question without resolving the Tur\'an problem.

\paragraph*{Motivation: Circuit Lower Bounds}

Let $F(n,t,k)$ denote the smallest number of $k$-CNF formulas required to cover $\THR_t$, i.e., the smallest number $q$ such that $\THR_t =\bigvee_{i=1}^q \psi_i$ for some $k$-CNF formulas $\psi_1, \dots, \psi_q$. 
$F(n,t,k)$ is the minimum fan-in of the top OR gate
in a $\Sigma\Pi\Sigma_k$ circuit that computes $\THR_t$.
Note that each $\psi_i$ must not accept assignments of weights less than $t$ -- we call such a formula $t$-admissible. We observe that $F(n,t,k)$ is almost tightly controlled by $S(n,t,k)$ and the analogously defined $S^{+}(n,t,k)$, the  maximum number of weight-$t$ assignments accepted by a \emph{monotone} $t$-admissible $k$-CNF.

\begin{observation}\label{obs:kCNFcovering}
	\[ \frac{\binom{n}{t}}{S(n,t,k)} \le  F(n,t,k) \le \left \lceil \frac{\binom{n}{t}}{S^{+}(n,t,k)} \cdot n \right\rceil. \] 
\end{observation}
\begin{proof}
	The first inequality is immediate, as there are $\binom{n}{t}$ weight-$t$ assignments to be covered, while any admissible $k$-CNF formula can accept at most $S(n,t,k)$ weight-$t$ assignments.

	To show the second inequality, we use the probabilistic method: Let $\psi$ be a monotone $t$-admissible $k$-CNF formula on $n$ variables maximizing the number of satisfying assignments of weight $t$. Let $D$ be the distribution of $k$-CNF formulas obtained from $\psi$ by taking a uniformly chosen random permutation of the variables. If $\psi'$ is sampled from $D$, the probability that a fixed assignment of weight $t$ satisfies $\psi'$ is   
$S^{+}(n,t,k)/\binom{n}{t}$.
	
	Let $\phi$ be the disjunction of $r:= \lceil (\binom{n}{t} / S^{+}(n,t,k)) n \rceil$  random $k$-CNF formulas chosen independently according to $D$. Then the probability that any fixed assignment of weight $t$ fails to satisfy any of the $r$ random $k$-CNFs is at most $(1-S^{+}(n,t,k)/\binom{n}{t})^r \le e^{-r \cdot S^{+}(n,t,k)/\binom{n}{t}} \le e^{-n}$. Thus, a union bound over all $\binom{n}{t}\le 2^n < e^n$ assignments shows that with positive probability, $\phi$ accepts all assignments of weight $t$. Observe that the monotonicity and $t$-admissibility of $\psi$ imply that $\phi$ is equivalent to $\THR_t$, establishing the claim. 
\end{proof}

If \cref{conjecture} is true, we obtain
\begin{align}\label{eq:F-lowerbound}
 F(n,t,k) \ge 2^{\frac{n}{k-1} (1/2-o(1)) (1-\alpha) \log(2 \pi (k-1) \alpha)},
\end{align}
which is proved in \cref{sec:intermediate}. The following consequence is immediate.

\begin{observation}\label{obs:F-lowerbound}
	If \cref{conjecture} holds, then for infinitely many $n$,  $\THR_{\alpha n}$ requires $\Sigma \Pi \Sigma_k$ circuits of size $2^{\Omega(\frac{n}{k} \log k)}$.
\end{observation}
\begin{proof}
	Since any $\Sigma \Pi \Sigma_k$ circuit can be viewed as a disjunction of $k$-CNF formulas, \eqref{eq:F-lowerbound} gives a lower bound on the number of gates in the middle layer of any $\Sigma \Pi \Sigma_k$ circuit.  
\end{proof}

As the best known  $\Sigma \Pi\Sigma_k$-circuit lower bounds are of the form $2^{\Omega(\frac{n}{k})}$, 
the lower bound in \cref{obs:F-lowerbound} would improve the asymptotic dependence in $k$, resolving an open question (Problem 3) of \cite{HastadJuknaPudlak_1995_cc}. In fact, we only need to show optimality for any threshold $t= \frac{n}{o(k)}$ to improve the asymptotic dependence.

For depth-3 circuits without bottom fan-in restriction, the best known lower bounds for Majority (i.e., $\THR_{n/2}$) are 
of the form $2^{\Omega(\sqrt{n})}$ \cite{KlawePaulPippengerYannakakis_1984_stoc, HastadJuknaPudlak_1995_cc, PaturiPudlakSaksZane_2005_jacm}. \cref{conjecture} would improve this to  $2^{\Omega(\sqrt{n \log n})}$, matching a construction by~\cite{Boppana_1984_stoc}. \footnote{Note that \cite{Wolfovitz_2006_ipl} incorrectly claims constructing depth 3 circuit for majority of size $2^{O(\sqrt{n})}$, which is due to the oversight explained in \cref{footnote:Wolfovitz}.} 

Finally, for fixed $k\ge 3$, the current best lower bounds for an explicit function are given by \cite{PaturiPudlakSaksZane_2005_jacm}, and are of the form $2^{\frac{\mu_k}{k-1} n -o(n)}$ where $\mu_k = \sum_{j \ge 1} \frac{1}{j(j+1/(k-1))}$ with  $\lim_{k\to \infty} = \frac{\pi^2}{6} \approx 1.645$. \cref{conjecture} would give stronger lower bounds for every $k\ge 12$.

\paragraph*{Related and Subsequent Works}

We mention that this work was completed in 2021. Subsequently, several works have studied combinatorial properties of $k$-CNFs computing threshold functions \cite{FGT22VC, Amano23Majority, GPPST24majority} as well as complexity of specialized circuits computing majority \cite{LRT22Majority}. We mention that all of these works except \cite{GPPST24majority} were done independent of this work. Also, \cite{Amano23Majority} showed \cref{conjecture} is false for $k=4$ by showing a block construction using blocks of size $2k$ gives a better construction. We believe the conjecture should hold for odd $k$ and a variant of it should hold for even $k$ so that \cref{obs:F-lowerbound} holds.

\paragraph*{Open Problems}

It remains open to prove or disprove the optimality of the 
adaptive block construction for $k\ge 3$ (\cref{conjecture}). 

Note that the adaptive block construction is natural: It decomposes the variables into blocks, where the block size is chosen as the largest number $b$ such that we can express $\THR_{\alpha b}$ perfectly as a $k$-CNF (to see this claim, note that \cref{thm:turan1} proves $S(b,b-k,k) < \binom{b}{b-k}$). However, it is unclear whether a larger block size might enable us to capture (while not \emph{all}) a relatively large fraction of weight-$\alpha b$ solutions in each block.
This suggests the following question: Can we show that there are optimal formulas with a block construction, i.e., a construction that 
partitions the variables into blocks $B_1, \dots, B_b$ with $b\ge 2$ and 
designs a $k$-CNF $\Psi_i$ for each block $B_i$ such that any 
satisfying assignment for $\Psi_i$ has Hamming weight at least $t_i$, where $t_1, \dots, t_b$ are such that $\sum_{i=1}^b t_i = t$?

It is also interesting to resolve whether non-monotone clauses are beneficial for capturing threshold functions.
While we are able to show that $S^{+}(n,t,k) = S(n,t,k)$ for $k=2$,
the question is open for $k\geq 3$.

\section{Preliminaries} \label{sec:preliminaries}

We set $[n] = \{1, \dots, n\}$. A CNF formula $F$ over the variables $x_1, \dots, x_n$ is a conjunction of \emph{clauses} $(X_{i_1} \lor \dots \lor X_{i_k})$, where each \emph{literal} $X_{i_j}$ is either a variable~$x_{i_j}$ or its negation~$\neg x_{i_j}$. Here, $k$ is called the \emph{width} of the clause. $F$ is said to be a $k$-CNF formula if all of its clauses have width at most $k$. Moreover, $F$ is called \emph{monotone} if all of its clauses only contain positive literals. We will often identify a CNF formula $F$ with the Boolean function $F : \{0, 1\}^n \to \{0, 1\}$ it computes. Throughout, we use lower-case letters $x, y, z$ to denote variables and upper-case letters $X, Y, Z$ to denote their corresponding literals.

For an \emph{assignment} $\alpha \in \{0, 1\}^n$, let $\wt(\alpha) = \sum_i \alpha_i$ denote the Hamming weight of~$\alpha$. Let $F : \{0, 1\}^n \to \{0, 1\}$ be a Boolean function. We set $\sat(F) = \{ \alpha : F(\alpha) = 1 \}$ and $\unsat(F) = \{0, 1\}^n \setminus \sat(F)$. We let $\sat_t(F)$ denote the subset of satisfying assignments with Hamming weight equal to~$t$. We also  say that a formula $F$ is $t$-\emph{admissible} if  $F$ does not satisfy any assignment with weight less than $t$. By definition, every formula is $0$-admissible.

We finally define the quantities which are the main focus of this paper. For~$k, t \leq n$, let:
\begin{itemize}
\item $S(n, t, k) = \max\{\, |\sat_t(F)| : \text{\normalfont$F$ is a $t$-admissible $k$-CNF over $n$ variables} \,\}$,
\item $S^+(n, t, k) = \max\{\, |\sat_t(F)| : \text{\normalfont$F$ is a $t$-admissible \emph{monotone} $k$-CNF over $n$ variables} \,\}$,
\item $F(n,t,k) = \min\{q : \THR_t = \bigvee_{i=1}^q \psi_i\}$ where $\psi_i$ are $k$-CNF formulas. 
\end{itemize}

\begin{lemma}[{{{\cite[Chapter 10,~Lemma 7]{MacWilliamsSloane_1977}}}}] \label{lem:binom-bounds}
For $\alpha \in (0, 1)$, let $H(\alpha) = -\alpha \log_2(\alpha) - (1-\alpha) \log_2(1-\alpha)$ denote the binary entropy function. Then, for all $n \geq 0$:
\begin{align*}
	\frac{2^{nH(\alpha)}}{\sqrt{8n\alpha(1-\alpha)}} \le \binom{n}{\alpha n} \le \frac{2^{nH(\alpha)}}{\sqrt{2\pi n\alpha(1-\alpha)}}.
\end{align*}
\end{lemma}

\section{Bounds on \texorpdfstring{\boldmath$S(n, t, k)$}{S(n, t, k)}} \label{sec:bounds}
In this section we propose $k$-CNF constructions for optimally 
capturing threshold functions.
Our constructions are block constructions: we group variables   
into  blocks and design a $k$-CNF on each block of variables to accept assignments with certain minimum Hamming weight.
We focus on the regime where $k$ is constant and $t$ may depend on $n$.
For small thresholds ($t\leq n/k$), we show that our constructions  are optimal and unique.
For intermediate thresholds ($t=\alpha n$ for $0<\alpha<1$), we show we can do better by varying block sizes as well as thresholds of the blocks using our adaptive block construction.
In \cref{sec:k2}, we will show that for $k=2$, our adaptive block construction is optimal for all thresholds $t$.

We will first explain how the block construction works
by showing that the blocks  can be  put together to create the final $k$-CNF. Subsequently, we design specific $k$-CNF for the blocks.


\begin{lemma}[Putting the block together] \label{lem:blocks}
Let $k, t_1, \dots, t_a, n_1, \dots, n_a$ be positive integers. We then have 
\begin{align*}
    S\left(\sum_{i=1}^a n_i, \sum_{i=1}^a t_i, k\right) \geq \prod_{i=1}^a S(n_i, t_i, k).
\end{align*}
\end{lemma}
\begin{proof}
Suppose that $G_i$ is a $t_i$-admissible $k$-CNF formula and $|\sat_{t_i}(G)| = S(n_i, t_i, k)$. Let $F = \bigwedge_{i=1}^a G_i$ where the subformulas $G_i$ have distinct input variables. $F$ has $\sum_i n_i$ input variables in total and it is easy to check that $\sat(F) = \sat(G_1) \times \dots \times \sat(G_a)$. Moreover, $F$ is $t$-admissible where $t := \sum_i t_i$ and $|\sat_t(F)| = \prod_i |\sat_{t_i}(G_i)| = \prod_i S(n_i, t_i, k)$.
\end{proof}

For most applications it suffices to apply \cref{lem:blocks} in the following simpler form, where all $t_i$ are the same and all $n_i$ are the same:
\begin{equation*}
    S(an, at, k) \geq S(n, t, k)^a.
\end{equation*}

\subsection{Bounds for Small Thresholds}
We first focus on the setting $t \leq n/k$. In this regime it is easy to prove optimal bounds for $S(n, t, k)$; these bounds first proved  in~\cite{Wolfovitz_2006_ipl}. We provide a stronger version of the same theorem where we show the uniqueness (up to isomorphism) of the optimal solution for $t\leq n/k$. 

\begin{theorem}[Small Thresholds] \label{thm:small-thresholds}
For $n \geq k \cdot t$, $S(n, t, k) = k^t$. Moreover, a $k$-CNF formula attains this bound only if it has a subset of $t$ disjoint clauses of the form $(x_1 \lor \dots \lor x_k)$.
\end{theorem}

\begin{proof}
It is easy to see that $S(n, t, k) \geq k^t$ by analyzing the claimed unique construction: Let $M(x_1, \dots, x_k) = (x_1 \lor \dots \lor x_k)$ be the $k$-CNF formula accepting all inputs with Hamming weight at least one, and observe that $M$ certifies the bound $S(k, 1, k) \geq k$. By \cref{lem:blocks} it follows that $S(n, t, k) \geq S(kt, t, k) \geq S(k, 1, k)^t = k^t$.

Next we prove that $S(n, t, k) \leq k^t$ by induction on $t$. The base case $t=0$ is trivial: There is only one assignment with Hamming weight $0$. For the inductive step assume that $t \geq 1$ and let~$F$ be a $t$-admissible optimal formula and $|\sat_t(F)| = S(n, t, k)$. Since the all-zero assignment does not satisfy $F$,  $F$ must contain a monotone clause, say $F = F' \land (x_1 \lor \dots \lor x_k)$. Observe that $F'$ is a $(t-1)$-admissible $k$-CNF formula. Indeed, whenever an assignment satisfies  $F'$ it suffices to switch on at most one additional variable amongst $x_1, \dots, x_k$ to satisfy $F$. 
Using distributive law on the last clause, we see that $F = \bigvee_{i=1}^k(F'\land x_i)$ and so, 
$$|\sat_t(F)| \le \sum_{i=1}^k|\sat_t(F'\land x_i)| \le \sum_{i=1}^k|\sat_{t-1}(F')| = k \cdot |\sat_{t-1}(F')|$$
It follows by the induction hypothesis that $|\sat_t(F)| \leq k \cdot |\sat_{t-1}(F')| \leq k^t$.

To prove uniqueness of the optimal construction, suppose that $|\sat_t(F)| = k^t$. Then the inequality $|\sat_t(F)| \leq k \cdot |\sat_{t-1}(F')| \leq k^t$ must be tight and we must have $|\sat_{t-1}(F')| = k^{t-1}$. Therefore, by the induction hypothesis~$F'$ contains $t-1$ disjoint monotone $k$-clauses $C_1, \dots, C_{t-1}$ and the weight $t-1$ satisfying assignments of $F'$ are exactly the assignments which switch on one variable from each clause $C_i$. Let $C_t := (x_1 \lor \dots \lor x_k)$; we claim that $C_t$ is disjoint from $C_1, \dots, C_{t-1}$. Indeed, if $C_t$ shares a variable, say $x_1$, with another clause $C_i$, then we can construct a satisfying assignment of $F$ with Hamming weight less than $t$ by switching on $x_1$ and one variable out of each clause $C_j, j \neq i$. 
\end{proof}

\subsection{Bounds for Intermediate Thresholds}
\label{sec:intermediate}
In this section we present our candidate construction for intermediate thresholds $t = \alpha n$ where $\alpha$ is a constant. We start with the following lemma, stating that for $t = n-k+1$ it is possible find a $t$-admissible formula that is satisfied by \emph{all} weight-$t$ assignments. In fact, $t = n-k+1$ is the smallest threshold with this property and any threshold below $n-k+1$ suffers from the difficulties of the Tur\'an problem (see \cref{sec:turan}).

\begin{lemma} \label{lem:n-k+1}
$S(n, n-k+1, k) = \binom{n}{k-1}$.
\end{lemma}
\begin{proof}
The upper bound is clear since there are only $\binom{n}{k-1}$ assignments with Hamming weight $n-k+1$. To show the lower bound, consider the $k$-CNF formula $F(x_1, \dots, x_n)$ containing a clause $C_I = \bigvee_{i \in I} x_i$ for all size-$k$ subsets $I \subseteq [n]$. On the one hand, we claim that $F$ is $(n-k+1)$-admissible. Indeed, for any assignment $\alpha$ with weight $\wt(\alpha) \leq n-k$ we can find a clause $C_I$ which is falsified by $\alpha$ by picking $I \subseteq \{ i : \alpha_i = 0\}$. On the other hand it is easy to see that every assignment with weight $n-k+1$ is satisfying for $F$ and thus $|\sat_{n-k+1}(F)| = \binom{n}{n-k+1} = \binom{n}{k-1}$.
\end{proof}

\begin{theorem}[Intermediate Thresholds] \label{thm:intermediate-thresholds}
Let $t = \alpha n$. Whenever $b = \frac{k-1}{1-\alpha}$ is an integer dividing~$n$ it holds that
\begin{equation*}
    S(n, t, k) \geq \binom{b}{\alpha b}^{n/b}.
\end{equation*}
\end{theorem}
We remark that an analogous more convoluted statement is true for the cases where $b = \frac{k-1}{1-\alpha}$ is not an integer by rounding in the appropriate places.
\begin{proof}
Note that if $b = \frac{k-1}{1-\alpha}$ is an integer, then so is $\alpha b = b - k + 1$. By \cref{lem:blocks} we have $S(n, t, k) = S(\frac nb \cdot b, \frac nb \cdot \alpha b, k) \geq S(b, \alpha b, k)^{n/b}$ and by \cref{lem:n-k+1}, $S(b, \alpha b, k) = S(b, b - k + 1, k) = \binom{b}{k-1} = \binom{b}{\alpha b}$. In combination these bounds prove the claim.
\end{proof}

Next, we prove that if \cref{thm:intermediate-thresholds} turns out to be optimal, then we obtain improved depth-$3$ circuit lower bounds as a consequence.

\begin{proposition}[Optimality Implies Circuit Lower Bounds]
Let $t = \alpha n$. If the bound in \cref{thm:intermediate-thresholds} is tight (i.e., it also holds that $S(n, t, k) \leq \binom{b}{\alpha b}{}^{n/b}$ for $b = \frac{k-1}{1-\alpha}$), then
\begin{equation*}
    F(n, t, k) \geq 2^{\frac{n}{k-1} \cdot \frac{1-\alpha}2 \log(2 \pi \alpha (k-1)) - o(n)} \geq 2^{\Omega(\frac nk \log k)}.
\end{equation*}
\end{proposition}

\begin{proof}
Assume that the bound in \cref{thm:intermediate-thresholds} is tight. Then, by \cref{obs:kCNFcovering} and using the approximation in \cref{lem:binom-bounds} we obtain
\begin{multline*}
    F(n, t, k)
    \geq \frac{\binom nt}{S(n, t, k)}
    = \frac{\binom n{\alpha n}}{\binom b{\alpha b}^{n/b}}
    \geq \frac{\sqrt{2\pi b \alpha (1-\alpha)}^{n/b} \cdot 2^{nH(\alpha)}}{\sqrt{8n \alpha (1-\alpha)} \cdot 2^{nH(\alpha)}} \\
    \geq \frac{\sqrt{2\pi \alpha(k-1)}^{n/b}}{2^{o(n)}}
    = 2^{\frac{n}{k-1} \cdot \frac{1-\alpha}2 \log(2 \pi \alpha (k-1)) - o(n)}.
\end{multline*}
\end{proof}

\section{Tight Bounds for \texorpdfstring{\boldmath$S(n, t, 2)$}{S(n, t, 2)}} \label{sec:k2}

We here show that for $k = 2$ and all $t$, our adaptive block construction is optimal.
We do this by starting with an arbitrary optimal formula and transforming it to a monotone formula whilst retaining its optimality.
For monotone formulas we show that the problem is equivalent to maximizing the number of maximal independent sets of a fixed size in a graph, a problem that was recently solved~\cite{SongYao_2020_arxiv}.
It follows from \cref{thm:small-thresholds} that for $t \leq n/2$ our constructions are unique and optimal. So in this section we show optimality for $t > n/2$. Our goal is to prove that
\begin{equation*}
	S(n, t, 2) \approx \left(\frac n{n-t}\right)^{n-t}
\end{equation*}
as detailed in the next theorem.

\begin{theorem}[Tight Bound for 2-CNFs] \label{thm:2-cnf-bound}
Write $n = (n-t)q + r$ where $q, r$ are the unique integers with $0 \leq r < n-t$. Then $S(n, t, 2) = q^{n-t-r}(q+1)^r$.
\end{theorem}

It is easy to establish the lower bound. To that end, we first observe that $S(n, n-1, 2) = n$ by setting $k=2$ in \cref{lem:n-k+1}.
We then apply \cref{lem:blocks}. Note that $t = (n-t-r) (q-1) + rq$ and $n = (n-t-r)q + r(q+1)$ and thus
\begin{equation*}
	S(n, t, 2) \geq S(q, q-1, 2)^{n-t-r} S(q+1, q, 2)^r \geq q^{n-t-r} (q+1)^r.
\end{equation*}

For the rest of this section we focus on the matching upper bound, starting with a definition. Let~$F$ be a $2$-CNF formula. The \emph{implication graph $G(F)$} is a directed graph with~$2n$ vertices $\{ x_1, \dots, x_n, \neg x_1, \dots, \neg x_n \}$. For each clause $X \lor Y$, where~$X$ and~$Y$ are literals, the implication graph contains two edges: $\neg X \rightarrow Y$ and~$\neg Y \rightarrow X$. We say that $F$ is \emph{acyclic} if its implication graph $G(F)$ is acyclic.

\begin{lemma}
There exists an acyclic $2$-CNF formula $F$ that is $t$-admissible and $|\sat_t(F)| = S(n, t, 2)$.
\end{lemma}
\begin{proof}
We start from a $2$-CNF formula $F$ that is $t$-admissible and $|\sat_t(F)| = S(n, t, 2)$. The goal is to to construct another formula $F'$ with the same properties and the additional guarantee that $G(F')$ is acyclic.

Suppose that there exists a cycle in $G(F)$, and let $X, Y$ be any two literals in that cycle. We can assume without loss of generality that the literal $X = x$ is positive (as for any cycle containing only negative literals, there exists a symmetric one containing only positive literals). Let $F_{x \gets Y}$ denote the formula obtained from $F$ by replacing every occurrence of the variable~$x$ by $Y$, and set $F' = F_{x \gets Y} \land (x \lor \neg Y)$. We argue that $F'$ satisfies the claimed guarantees:
\begin{itemize}
\item $F'$ is $t$-admissible: Suppose for the sake of contradiction that there exists a satisfying assignment $\alpha \in \sat(F')$ with weight less than~$t$. Since $F$ only accepts assignments with weight at least~$t$ and $F_{x \gets Y}(\alpha) = 1$, $\alpha$~must be a weight $(t-1)$ assignment which assigns $x = 0$. But since $F'(\alpha) = 1$ it also follows that $\alpha$ assigns $Y = 0$. Hence $F(\alpha) = F_{x \gets Y}(\alpha) = 1$, which is a contradiction.
\item $|\sat_t(F')| = S(n, t, 2)$: Let $\alpha$ be any satisfying assignment of $F$. Since $x, Y$ are contained in a cycle in $G(F)$, it must hold that $x$ and $Y$ are assigned the same truth value under $\alpha$. Hence $F_{x \gets Y}(\alpha) = F(\alpha) = 1$ and thus $F'(\alpha) = 1$. 
\end{itemize}
Notice that $x$ is not part of a cycle in $G(F')$ and thus the number of variables contained in cycles has decreased. Therefore, we can repeat this construction until the implication graph becomes acyclic.
\end{proof}

\begin{lemma} \label{lem:k2-monotone}
There exists an acyclic, monotone $2$-CNF formula $F$ that is $t$-admissible and $|\sat_t(F)| = S(n, t, 2)$.
\end{lemma}
\begin{proof}
Using the previous lemma we can assume that $F$ is a $2$-CNF formula such that
\begin{enumerate}
\item $F$~is acyclic,
\item $F$ is $t$-admissible, and
\item $|\sat_t(F)| = S(n, t, 2)$.
\end{enumerate}
The goal is to turn $F$ into a monotone formula $F'$ while maintaining the three properties.

Suppose that $F$ is not monotone. Then there exists a clause $F = F'' \land (\neg x \lor Y)$ (for some variable $x$ and some literal $Y$). We set
\begin{equation*}
	F' = F'' \land \left(\bigwedge_{Z \in \predecessors(x)} (\neg Z \lor Y)\right) \land \left(\bigwedge_{W \in \successors(Y)} (\neg x \lor W) \right)\!,
\end{equation*}
where $\predecessors(\cdot)$ and $\successors(\cdot)$ denote the predecessors and successors of a vertex in $G(F)$, respectively. In other words, $F'$ is obtained from $F$ by deleting the edge $(x, Y)$ in the implication graph of $F$ and adding edges from the predecessors of $x$ to $Y$ and edges from $x$ to the successors of $Y$; see \cref{fig:implication-graph} for an illustration.

Clearly $F'$ is acyclic and any satisfying assignment of $F$ is also satisfying for $F'$, thus $|\sat_t(F')| = S(n, t, 2)$. It remains to check that $F'$ is $t$-admissible. Suppose that $\alpha$ is a satisfying assignment of $F'$ with weight less than~$t$. Since $\alpha \in \unsat(F)$ and the only clause in $F$ which is missing in $F'$ is $(\neg x \lor Y)$, it must hold that $\alpha$ falsifies that clause, i.e., $\alpha(x) = 1$ and $\alpha(Y) = 0$. Moreover, by the construction of $F'$ we have $\alpha(Z) = 0$ for all $Z \in \predecessors(x)$ and $\alpha(W) = 1$ for all $W \in \successors(Y)$.

Let $\beta$ be the assignment obtained from $\alpha$ by flipping~$x$ and~$Y$, i.e., $\beta(x) = 0$, $\beta(Y) = 1$ and $\beta(Z) = \alpha(Z)$ for all other $Z$. We can check that $\beta$ satisfies $F$: Clearly $(\neg x \lor Y)$ is satisfied, so we only have to check that $F''(\beta) = 1$. Recall that $F''(\alpha) = 1$ and so we only have to check for the implications including $x$ or $Y$. We only focus on $x$; the argument for $Y$ is symmetric. Since $\beta(Z) = \alpha(Z) = 0$ for all predecessors $Z$ of $x$, all implications leading to $x$ are satisfied. Moreover, as $\beta(x) = 0$, all outgoing implications are satisfied as well.

Thus $F(\beta) = 1$, yielding a contradiction: Since $\beta$ is obtained from $\alpha$ by flipping the variable~$x$ and the literal~$Y$ we have $\wt(\beta) \leq \wt(\alpha) < t$, which contradicts the assumption that $F$ is $t$-admissible.

Finally, we argue that repeatedly applying the previous transformation will result in a purely monotone formula. Since $G(F)$ is acyclic, there exists a topological ordering $\iota : V(G(F)) \to [2n]$ (that is, $\iota(X) < \iota(Y)$ whenever there is an edge from $X$ to $Y$). We define the \emph{length} of an edge $(X, Y) \in E(G(F))$ by $\iota(Y) - \iota(X)$. It is easy to see that in the transformation from $F$ to $F'$ we remove an edge $(X, Y)$ from $G(F)$ and add a number of edges with strictly larger length. As the maximum edge length is bounded (by $2n$), this process cannot continue indefinitely and we eventually reach a formula which is monotone.
\end{proof}

\begin{figure}[t]
\begin{center}
\begin{tikzpicture}[
    vertex/.style={
        draw,
        circle,
        minimum size=.5cm,
        semithick,
        large/.style={minimum size=.7cm}
    },
    edge/.style={
        ->,
        >=latex,
        shorten >=.15cm,
        shorten <=.15cm,
        semithick,
        deleted/.style={dashed},
        added/.style={
            dashed,
            color=lipicsYellow
        },
    },
    region/.style={
        color=lipicsGray,
        thick,
        dashed,
        dash pattern=on .16cm off .16cm,
        rounded corners
    }
]
\node[vertex, large] (x) at (-1, 0) {$x$};
\node[vertex, large] (Y) at (1, 0) {$Y$};

\node[vertex] (Z1) at (-3, 1.2) {};
\node[vertex] (Z2) at (-3, 0.4) {};
\node[vertex] (Z3) at (-4, -0.4) {};
\node[vertex] (Z4) at (-3, -1.2) {};
\draw[edge] (Z1) to (x);
\draw[edge] (Z2) to (x);
\draw[edge] (Z2) to (Z4);
\draw[edge] (Z3) to (Z4);
\draw[edge] (Z4) to (x);

\node[vertex] (W1) at (3, 1.2) {};
\node[vertex] (W2) at (4, 0.4) {};
\node[vertex] (W3) at (3, -0.4) {};
\node[vertex] (W4) at (3, -1.2) {};
\draw[edge] (Y) to (W1);
\draw[edge] (Y) to (W3);
\draw[edge] (W3) to (W2);
\draw[edge] (Y) to (W4);
\draw[edge] (W4) to (W2);

\draw[region] (-2.5, 1.7) -- (-2.5, -1.7) -- node[label, below] {$\predecessors(x)$} (-4.5, -1.7) -- (-4.5, 1.7) -- cycle;
\draw[region] (2.5, 1.7) -- (2.5, -1.7) -- node[label, below] {$\successors(Y)$} (4.5, -1.7) -- (4.5, 1.7) -- cycle;

\draw[edge, added, bend right=44] (Z1) to (Y);
\draw[edge, added, bend right=39] (Z2) to (Y);
\draw[edge, added, bend right=28] (Z3) to (Y);
\draw[edge, added, bend right=23] (Z4) to (Y);

\draw[edge, added, bend left=23] (x) to (W1);
\draw[edge, added, bend left=28] (x) to (W2);
\draw[edge, added, bend left=39] (x) to (W3);
\draw[edge, added, bend left=44] (x) to (W4);

\draw[edge, deleted] (x) to node[label, above] {$\neg x \lor Y$} (Y);

\end{tikzpicture}
\end{center}
\caption{Illustrates the implication graph of $F$ in the proof of \cref{lem:k2-monotone}. In the transformation from $F$ to $F'$, the black dashed edge is deleted and the yellow dashed edges are added.} \label{fig:implication-graph}
\end{figure}
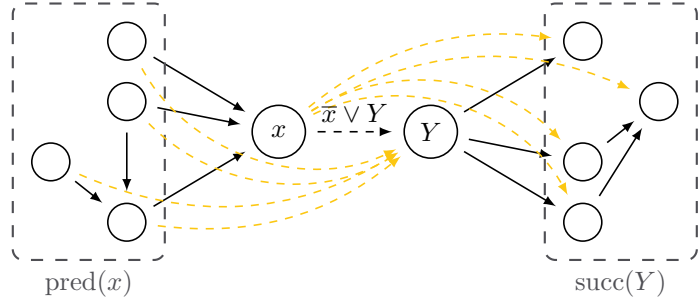

As a final ingredient, we need the following recent result by Song and Yao~\cite{SongYao_2020_arxiv}:

\begin{lemma}[Maximum Number of Maximal Independent Sets~{{{\cite{SongYao_2020_arxiv}}}}] \label{lem:maximal-independent-sets}
Let $n = (n-t)q + r$ where~$q, r$ are the unique integers with $0 \le r < n-t$. Then every $n$-vertex graph has at most $q^{n-t-r}(q+1)^r$ maximal independent sets of size $n - t$.
\end{lemma}

\begin{proof}[Proof of \cref{thm:2-cnf-bound}]
We have already proven the lower bound, so it suffices to show that $S(n, t, 2) \leq q^{n-t-r} (q+1)^r$. By \cref{lem:k2-monotone} we can assume that there exists a monotone $2$-CNF formula $F$ that is $t$-admissible and $|\sat_t(F)| = S(n, t, 2)$. We interpret $F$ as a graph $G = (V, E)$ as follows. Let $V$ be the variables of $F$ and add an edge $(x, y)$ to~$G$ if the clause $(x\lor y)$ appears in $F$.

Observe that any satisfying assignment for $F$ corresponds to a vertex cover of $G$. Moreover, there is a one-to-one correspondence between satisfying assignments with weight $t$ and \emph{minimal} vertex covers of size $t$ in $G$, since $F$ is $t$-admissible. In turn, it is well known that these correspond to maximal independent sets of size $n-t$ in $G$. Thus, the number of satisfying assignments with weight $t$ for any formula $F$ is bounded by the maximum number of maximal independent sets of size $n-t$ in $n$-vertex graphs. Finally, we apply \cref{lem:maximal-independent-sets}.
\end{proof}

\section{Connection to the Tur\'an Problem} \label{sec:turan}
In this section we prove that for certain parameters with $k\ge 3$, our problem of bounding $S(n, t, k)$ coincides with the famously open \emph{Tur\'an Problem}. We start with a formal definition of the Tur\'an Problem. 

A \emph{$k$-set system} over some universe $U$ is a collection $\mathcal S = \{S_1, \dots, S_m\}$ of size-$k$ subsets of~$U$. We refer to $m$ as the \emph{size} of $\mathcal S$.

\begin{definition}[Tur\'an and Covering numbers]
Let $k \leq q \leq n$ be nonnegative integers.
\begin{itemize}
\item The \emph{Tur\'an number $T(n, q, k)$} is the size of the smallest $k$-set system $\mathcal S$ over the universe~$[n]$ such that every size-$q$ subset of $[n]$ contains at least one set $S \in \mathcal S$.
\item The \emph{covering number $C(n, q, k)$} is the size of the smallest $q$-set system $\mathcal S$ over the universe~$[n]$ such that every size-$k$ subset of $[n]$ is contained in at least one set $S \in \mathcal S$.
\end{itemize}
\end{definition}

The task of  determining Tur\'an numbers $T(n, q, k)$ is called the \emph{Tur\'an Problem}. This problem, and also the problem of determining 
the covering numbers $C(n, q, k)$, received considerable attention in the literature \cite{Sidorenko_1995_gc, Keevash_2011} and both are still far from being resolved. We show that for $t = n-k$, determining  $S(n, t, k)$ is equivalent to the  Tur\'an problem $T(n,k+1,k)$. This shows why proving general optimality for all $t$, when $k\ge 3$, like we did in \cref{sec:k2} is hard.

\begin{theorem} \label{thm:turan}
For all $k \leq n$:
\begin{equation*}
	T(n, k+1, k) = C(n, n-k, n-k-1) = \tbinom nk - S(n, n-k, k).
\end{equation*}
\end{theorem}

To prove \cref{thm:turan}, we first establish two facts about optimal $k$-CNF formulas for $S(n, n-k, k)$: We will prove that there exists an optimal \emph{monotone} $k$-CNF formula (\cref{lem:large-threshold-monotone}) where all the clauses have width exactly equal to $k$ (\cref{lem:large-threshold-width-k}). As a consequence we can view an optimal formula for $S(n, n-k, k)$ as a set system, which can be related to $T(n, k+1, k)$ and $C(n, n-k, n-k-1)$.

\begin{lemma} \label{lem:large-threshold-monotone}
For $n \geq k$, there exists a monotone $(n-k)$-admissible $k$-CNF formula $F$ 
where $|\sat_{n-k}(F)| = S(n, n-k, k)$.
\end{lemma}
\begin{proof}
Let $F$ be an arbitrary $(n-k)$-admissible $k$-CNF formula where  $|\sat_{n-k}(F)| = S(n, n-k, k)$. We will show how to transform~$F$ into formula $F'$ so that
\begin{enumerate}
\item $F'$ is $(n-k)$-admissible,\label{lem:large-threshold-monotone:itm:level}%
\item $\sat_{n-k}(F) \subseteq \sat_{n-k}(F')$, and\label{lem:large-threshold-monotone:itm:sat}%
\item $F'$ contains strictly fewer non-monotone clauses than $F$.\label{lem:large-threshold-monotone:itm:monotone}%
\end{enumerate}
The lemma follows by repeating this construction until $F'$ is monotone.

Let $C = (x_1 \lor \dots \lor x_p \lor \neg y_1 \lor \dots \lor \neg y_q)$ be a non-monotone clause in $F$, and let $z_1, \dots, z_{n-p-q}$ denote the remaining variables in $F$. Moreover, let $F''$ denote the formula obtained from $F$ after removing~$C$. Let 
\begin{equation*}
	F' = F'' \land \bigwedge_{\substack{S \subseteq [n-p-q]\\|S| = k-p}} \left(\left(\bigvee_{i \in \text{\raisebox{0pt}[0pt][0pt]{$[p]$}}} x_i\right) \lor \left(\bigvee_{i \in S} z_i\right)\right).
\end{equation*}
Clearly $F'$ contains fewer non-monotone clauses than $F$, satisfying \cref{lem:large-threshold-monotone:itm:monotone}.

To show that \cref{lem:large-threshold-monotone:itm:level} holds, suppose that there exists an assignment $\alpha$ of weight less than $n-k$ which satisfies $F'$. Since $F(\alpha) = 0$, $\alpha$~must falsify $C$ and thus assigns $x_i = 0$ and $y_j  = 1$ for all $i, j$. Since $\wt(\alpha) < n - k$ there exist at least $(n - p - q) - (\wt(\alpha) - q) > k - p$ variables $z_i$ which are set to false under~$\alpha$, which guarantees that $F'$ is falsified.

Finally, we show that \cref{lem:large-threshold-monotone:itm:sat} also holds. Let $\alpha \in \sat_{n-k}(F)$; we show that $\alpha \in \sat_{n-k}(F')$. Clearly~$\alpha$ satisfies all clauses in $F''$, so we only have to focus on the new clauses. Observe that $\alpha$ also satisfies~$C$, hence there are only two cases: Either $\alpha$ assigns $x_i = 1$ for some $i$, or $\alpha$ assigns $x_i = 0$ for all $i$ and $y_j = 0$ for some $j$. In the former case every new clause is directly satisfied. In the latter case, $\alpha$ switches off at most $(n-\wt(\alpha)) - (p+1) = k - p - 1$ variables $z_i$, and thus all new clauses are satisfied as well. 
\end{proof}

\begin{lemma} \label{lem:large-threshold-width-k}
Let $n \geq k+t$ and let $F$ be a monotone $t$-admissible $k$-CNF formula such that $|\sat_t(F)| = S(n, t, k)$. If no clause in $F$ is redundant (i.e., we cannot remove any clause without violating $t$-admissibility of $F$), then all clauses in $F$ have width exactly $k$.
\end{lemma}
\begin{proof}
Suppose that $F$ contains a clause of width $k' < k$, say $F = F' \land (x_1 \lor \dots \lor x_{k'})$. Since no clause in $F$ is redundant there exists an assignment $\alpha$ with weight $\wt(\alpha) < t$ such that $F(\alpha) = 0$ and $F'(\alpha) = 1$. We can turn $\alpha$ into a satisfying assignment of $F$ by switching on one of the variables $x_1, \dots, x_{k-1}$; all clauses in $F'$ remain satisfied since we assume that $F$ is monotone. It follows that $\wt(\alpha) = t-1$. Let $y_1, \dots, y_{t-1}$ denote the variables turned on by $\alpha$: We are assuming that $n \geq k+t > k'+t$ and thus there exists another fresh variable $y_t \not\in \{x_1, \dots, x_{k'}, y_1, \dots, y_{t-1}\}$. Let $\beta$ be the assignment which turns on $y_t$ and otherwise leaves $\alpha$ unchanged. Then $\wt(\beta) = t$ and $F(\beta) = 0$ as the designated clause is not satisfied.

Let $G$ be the $k$-CNF formula obtained from $F'$ by adding all clauses $(x_1 \lor \dots \lor x_{k'} \lor y_j)$ for $j \in [t]$; clearly $\sat(F) \subseteq \sat(G)$. Notice that $\beta$ is satisfying for $G$ as $\beta$ satisfies all of the added clauses and also $F'(\beta) = F'(\alpha) = 1$. We will finally prove that $G$ is $t$-admissible which then contradicts the optimality of $F$. Indeed, for any assignment $\gamma$ with weight less than~$t$ we either have $F'(\gamma) = 0$ (in which case also $G(\gamma) = 0$) or the only clause in $F$ which is falsified is $(x_1 \lor \dots \lor x_{k'})$. In that case $\gamma$ is also falsifying for $G$ as it cannot turn on all $t$ variables $y_1, \dots, y_t$.
\end{proof}

\begin{proof}[Proof of \cref{thm:turan}]
It is elementary to check that $T(n, k+1, k) = C(n, n-k, n-k-1)$. Hence it suffices to prove that $C(n, n-k, n-k-1) = \binom nk - S(n, n-k, k)$.

We start with the upper bound. Let $F$ be a $(n-k)$-admissible $k$-CNF formula such that  $|\sat_{n-k}(F)| = S(n, n-k, k)$. By \cref{lem:large-threshold-monotone,lem:large-threshold-width-k} we can assume that $F$ is monotone and contains only clauses of width exactly $k$. Let $C_1, \dots, C_m$ denote the clauses of $F$ and view each clause as a size-$k$ set $C_i \subseteq [n]$ in the obvious way. Observe that each clause rejects a single assignment of weight $n-k$ and therefore, $S(n, n-k, k) = \binom n{n-k} - m = \binom nk - m$.

Next, consider the $(n-k)$-set system $\mathcal S = \{S_1, \dots, S_m\}$ where $S_i = [n] \setminus C_i$. We claim that every size-$(n-k-1)$ subset $T \subseteq [n]$ is contained in some set $S_i$. Indeed, let~$\alpha$ be the assignment which turns on all variables in $T$. Since $F$ is $(n-k)$-admissible, $\alpha$ must be falsifying for~$F$ and thus there exists a clause $C_i$ which is disjoint from $T$. It follows that $T \subseteq [n] \setminus C_i$. Finally, by the definition of covering numbers, $C(n, n-k, n-k-1) \leq m = \binom nk - S(n, n-k, k)$.

The proof for the other direction is analogous. Given an optimal set system $\mathcal S = \{S_1, \dots, S_m\}$ for the covering number $m = C(n, n-k, n-k-1)$, we construct a $k$-CNF formula containing the monotone clauses $C_1, \dots, C_m$ with $C_i = [n] \setminus S_i$. It is easy to see that $F$ is $t$-admissible for $t \geq n-k$ and thus $S(n, n-k, k) \geq \binom nk - C(n, n-k, n-k-1)$.
\end{proof}

\section{Connection to Steiner Systems} \label{sec:Steiner}
In this section we prove that, for certain parameters, the problem of constructing $t$-admissible $k$-CNF formulas which capture $S(n, t, k)$ assignments of weight $t$ coincides with that of constructing Steiner systems -- a special kind of designs.

Recall that a \emph{$k$-set system} over a universe $U$ is a collection $\mathcal{S}= \{S_1, \dots, S_m\}$ of size-$k$ subsets of~$U$. We refer to $m$ as the \emph{size} of $\mathcal S$.

\begin{definition}[Steiner System]
Let $r \leq q \leq n$ be nonnegative integers.
A \emph{Steiner system $(n, q, r)$} is a $q$-set system $\mathcal{S}$ over the universe~$[n]$ such that every size-$r$ subset of $[n]$ is contained in exactly one set $S \in \mathcal{S}$.
\end{definition}

The existence and explicit construction of Steiner systems has received considerable attention in the design literature \cite{ColbournDinitz_2006}. In fact, only recently the asymptotic existence of such systems was shown when $r$ and $q$ are constants~\cite{Keevash_2019_arxiv}.

Steiner set systems are related to optimal $t$-admissible $k$-CNF formulas when $n-k$ and $t$ are constants.

\begin{theorem}\label{thm:steiner}
Let $n > k+t$.
If Steiner systems with parameters $(n, n-k, t-1)$ exist, then $S(n,t,k) = \frac{k}{t} \cdot \binom{n}{t-1}$.
\end{theorem}
\begin{proof}[Proof. The Lower Bound]
We start proving the lower bound $S(n, t, k) \geq \frac kt \cdot \binom{n}{t-1}$. Let $\mathcal{S} = \{S_1, \dots, S_m\}$ be a Steiner system with parameters $(n,n-k,t-1)$. We construct a monotone formula $F$ the over variables $x_1, \dots, x_n$ as follows. For each set $S_i \in \mathcal{S}$, create a clause $$C_i = \bigvee_{j\in [n]\setminus S_i}x_j$$ and let $F = \bigwedge_i C_i$. Note that $F$ has width $k$ as $|S_i| = n-k$. Consider an arbitrary assignment~$\alpha$ such that $\wt(\alpha) = t-1$. Let $A$ be the set of variables that are set to $1$ by~$\alpha$. Since $\mathcal{S}$ is a Steiner system, there exists a unique set $S_i$ which covers $A$. This implies that~$C_i$ is not satisfied by $\alpha$, and thereby that $F$ is $t$-admissible.

We now show that $|\sat_t(F)| \geq \frac{k}{t} \cdot \binom{n}{t-1}$. Consider the (undirected) bipartite graph $G_t$ with vertex parts $V$ and $W$ defined by
\begin{itemize}
\item $V = \{ \alpha \in \{0, 1\}^n : \text{$\wt(\alpha) = t$ and $\alpha \in \unsat(F)$} \}$,
\item $W = \{ \beta \in \{0, 1\}^n : \wt(\beta) = t - 1 \}$.
\end{itemize}
We connect an assignment $\alpha\in V$ to an assignment $\beta \in W$ if $\alpha$ is obtained by flipping one of the zero bits of $\beta$. On the one hand, it is obvious that each vertex in $V$ has degree exactly $t$. On the other hand, each vertex $\beta \in W$ has degree exactly $n-t-k+1$ since there exists a unique clause not satisfied by $\beta$, and thus $\beta$ is adjacent to exactly $k$ satisfying assignments of weight~$t$. Hence, by double counting the edges in $G$ we get that $(n-t-k+1) \cdot \binom{n}{t-1} = t \cdot (\binom{n}{t} - |\sat_t(F)| )$ which proves the theorem.

\proofsubparagraph*{The Upper Bound.}
Next, we prove the upper bound $S(n, t, k) \leq \frac kt \cdot \binom{n}{t-1}$. Let $F$ be a $t$-admissible $k$-CNF formula with $|\sat_{t}(F)| = S(n, t, k)$ and let $\alpha$ be an assignment of weight $t-1$. Then, as $F(\alpha) = 0$, there exists a clause $C$ which is not satisfied by $\alpha$. Let $x_1, \dots, x_m$ denote the variables which are both not present in $C$ and turned off by $\alpha$. Then $m \ge n - k - (t-1) = n - (k+t-1)$. Now consider the assignments $\beta_1, \dots, \beta_m$ where $\beta_i$ turns on $x_i$ and assigns the rest of the variables are assigned as $\alpha$. It is immediate that $\wt(\beta_i) = t$ and $F(\beta_i) = 0$ for all $i$.

As this is true for all assignments $\alpha$ with weight $t-1$, we conclude that the number of rejected assignments is at least $\frac{n - (k+t-1)}t \binom{n}{t-1}$. Here the division by $t$ is necessary as we overcounted every assignment exactly $t$ times. It follows that
\begin{equation*}
    S(n, t, k) \le \binom{n}{t} - \frac{n - (k+t-1)}{t}\binom{n}{t-1} = \frac{k}{t} \cdot \binom{n}{t-1}. \qedhere
\end{equation*}
\end{proof}

An interesting observation about the above proof is that any formula with $S(n, t, k) = \frac{k}{t} \cdot \binom{n}{t-1}$  would in fact yield constructions of Steiner systems.
Only recently, the existence of Steiner systems $(n, q, r)$ for constants $q, r$ and sufficiently large $n$ satisfying the following divisibility condition: $\binom{q-i}{r-i} \textrm{ divides } \binom{n-i}{r-i} \textrm{ for } 0\le i\le r-1$ have been proven \cite{Keevash_2019_arxiv}.
This gives us the following corollary:

\begin{corollary}
Let $n > k+t$ where $t$ and $n-k$ are constants and $n$ is large enough and such that $\binom{n-k-i}{t-1-i}$ divides $\binom{n-i}{t-1-i}$ for $0\le i\le t-2$.
Then, $S(n, t, k) = \binom{n}{t-1}\cdot\frac{k}{t}$.
\end{corollary}

\bibliographystyle{plain}
\bibliography{refs}{}

\end{document}